\documentclass[preprint,sort&compress]{elsarticle}

\usepackage{graphicx}
\usepackage{array} 
\usepackage{amssymb}
\usepackage{amsthm}
\usepackage{amsmath}
\usepackage{url} 
\usepackage{enumitem}%项目编号
\usepackage{caption} %改变图表标题
\usepackage{subfigure}
\usepackage{txfonts} %默认字体times new roman
\usepackage{booktabs} %调整表格线与上下内容的间隔
\usepackage{longtable}
\usepackage{overpic}
\usepackage{multirow} %多行合并
\usepackage{natbib}
\setcitestyle{authoryear,round}
\usepackage{lineno}
\usepackage[scale=0.8]{geometry}
\usepackage[flushmargin,hang]{footmisc} %footnote
\usepackage{pgfplots}
\usepackage{filecontents}
\usepackage{titlesec}
\usepackage{epstopdf}

\newtheoremstyle{dotless}{}{}{\itshape}{}{\bfseries}{}{1em}{} % define the theorem lemma definition style 
\theoremstyle{dotless}
\newtheoremstyle{theoremdd}% name of the style to be used
{}% measure of space to leave above the theorem. E.g.: 3pt
{}% measure of space to leave below the theorem. E.g.: 3pt
{}% name of font to use in the body of the theorem
{}% measure of space to indent
{}% name of head font
{}% punctuation between head and body
{ }% space after theorem head; " " = normal interword space
{\thmname{#1}\thmnumber{ #2'}\thmnote{ (#3)}.}
  
\newtheorem{theorem}{Theorem}
\newtheorem{lemma}{Lemma}
\newtheorem{definition}{Definition}
\newtheorem{corollary}{Corollary}

\usepackage{hyperref}
\hypersetup{
    colorlinks=true,
    linkcolor=blue,
    filecolor=gray,      
    urlcolor=blue,
    citecolor=blue,
}

\makeatletter 

\let\@afterindenttrue\@afterindentfalse

\usepackage{etoolbox} %href only the year
\patchcmd{\NAT@citex}
  {\@citea\NAT@hyper@{%
     \NAT@nmfmt{\NAT@nm}%
     \hyper@natlinkbreak{\NAT@aysep\NAT@spacechar}{\@citeb\@extra@b@citeb}%
     \NAT@date}}
  {\@citea\NAT@nmfmt{\NAT@nm}%
   \NAT@aysep\NAT@spacechar\NAT@hyper@{\NAT@date}}{}{}

\patchcmd{\NAT@citex}
  {\@citea\NAT@hyper@{%
     \NAT@nmfmt{\NAT@nm}%
     \hyper@natlinkbreak{\NAT@spacechar\NAT@@open\if*#1*\else#1\NAT@spacechar\fi}%
       {\@citeb\@extra@b@citeb}%
     \NAT@date}}
  {\@citea\NAT@nmfmt{\NAT@nm}%
   \NAT@spacechar\NAT@@open\if*#1*\else#1\NAT@spacechar\fi\NAT@hyper@{\NAT@date}}
  {}{}
  
\renewenvironment{proof}[1][\proofname]{\par
  \pushQED{\qed}%
  \normalfont \topsep6\p@\@plus6\p@\relax
  \trivlist
  \item[\hskip\labelsep
        \bfseries
    #1\@addpunct{}\enspace]\ignorespaces% DELETED
}{%
  \popQED\endtrivlist\@endpefalse
}

\def\keyword{%
  \def\sep{\unskip, }%
 \def\MSC{\@ifnextchar[{\@MSC}{\@MSC[2000]}}
  \def\@MSC[##1]{\par\leavevmode\hbox {\it ##1~MSC:\space}}%
  \def\PACS{\par\leavevmode\hbox {\it PACS:\space}}%
  \def\JEL{\par\leavevmode\hbox {\it JEL:\space}}%
  \global\setbox\keybox=\vbox\bgroup\hsize=\textwidth
  \normalsize\normalfont\def\baselinestretch{1}
  \parskip\z@
  \noindent\textbf{Keywords: }
  \raggedright                         % Keywords are not justified.
  \ignorespaces}
\def\endkeyword{\par \egroup}

\def\@author#1{\g@addto@macro\elsauthors{\normalsize%
    \def\baselinestretch{1}%
    \upshape\authorsep#1\unskip\textsuperscript{%
      \ifx\@fnmark\@empty\else\unskip\sep\@fnmark\let\sep=,\fi
      \ifx\@corref\@empty\else\unskip\sep\@corref\let\sep=,\fi
      }%
    \def\authorsep{\unskip { and} \space}%
    \global\let\@fnmark\@empty
    \global\let\sep\@empty}%
    \@eadauthor={#1}
}

\renewenvironment{abstract}{\global\setbox\absbox=\vbox\bgroup
  \hsize=\textwidth\def\baselinestretch{1}%
  \noindent\unskip%\textbf{Abstract}
 \par\noindent\unskip\ignorespaces} %leftskip
 {\egroup}
 
\long\def\pprintMaketitle{\clearpage
  \iflongmktitle\if@twocolumn\let\columnwidth=\textwidth\fi\fi
  \resetTitleCounters
  \def\baselinestretch{1}%
  \printFirstPageNotes
  \begin{flushleft}%center
 \thispagestyle{pprintTitle}%
   \def\baselinestretch{1}%
    \LARGE\@title\par\vskip15pt
    \Large\elsauthors\par\vskip10pt
    \itshape\elsaddress\par\vskip-3pt
    \vskip5pt
    \ifvoid\absbox\else\unvbox\absbox\par\vskip10pt\fi
    \ifvoid\keybox\else\unvbox\keybox\par\vskip10pt\fi
    \hrule
    \vskip45pt
    \end{flushleft}%
  \gdef\thefootnote{\arabic{footnote}}%
  }

\gdef\emailauthor#1#2{\stepcounter{ead}%
     \g@addto@macro\@elseads{\raggedright%
      \let\corref\@gobble
      \eadsep\texttt{#1} \def\eadsep{\unskip,\space}}%(#2),remove the ba
}

 \def\ps@pprintTitle{%
     \let\@oddhead\@empty
     \let\@evenhead\@empty
     \def\@oddfoot{\footnotesize\itshape}
     \let\@evenfoot\@oddfoot}
     
     \def\printFirstPageNotes{%
  \iflongmktitle
   \let\columnwidth=\textwidth\fi
  \ifx\@tnotes\@empty\else\@tnotes\fi
  \ifx\@nonumnotes\@empty\else\@nonumnotes\fi
  \ifx\@cornotes\@empty\else\@cornotes\fi
  \ifx\@elseads\@empty\relax\else
  \let\thefootnote\relax
   \noindent\footnotetext{\ifnum\theead=1\relax
      {Email:\space}\else
      {Email:\space}\fi
     \@elseads
     }\fi
  \ifx\@elsuads\@empty\relax\else
   \let\thefootnote\relax
   \footnotetext{\textit{URL:\space}%
     \@elsuads}\fi
  \ifx\@fnotes\@empty\else\@fnotes\fi
  \iflongmktitle\if@twocolumn
   \let\columnwidth=\Columnwidth\fi\fi
}

\makeatother

\journal{Journal of the Operational Research Society}

\DeclareCaptionLabelFormat{mylabel}{#1 #2\hspace{0em}}
\newcommand{\tabincell}[2]{\begin{tabular}{@{}#1@{}}#2\end{tabular}}

\begin{document}
\begin{frontmatter}

\title{Capital flow constrained lot sizing problem with loss of goodwill and loan}

\author{Zhen Chen}
%\ead{chen\_zhen@buaa.edu.cn}
\author{Ren-qian Zhang\corref{cor1}}
\ead{zhangrenqian@buaa.edu.cn}

\cortext[cor1]%{Correspondence: Ren-qian Zhang, School of Economics and Management, Beihang University, Beijing 100191, China}
{\emph{Correspondence: Ren-qian Zhang, School of Economics and Management, Beihang University, Beijing 100191, China}}

\address{School of Economics and Management, Beihang University, Beijing 100191, China}

\begin{abstract}
We introduce capital flow constraints, loss of good will and loan to the lot sizing problem. Capital flow constraint is different from traditional capacity constraints: when a manufacturer launches production, its present capital should not be less than its present total production cost; otherwise, it must decrease production quantity or suspend production. Unsatisfied demand in one period may cause customer's demand to shrink in the next period considering loss of goodwill. Fixed loan can be adopted in the starting period for production. A mixed integer model for a deterministic single-item problem is constructed. Based on the analysis about the structure of optimal solutions, we approximate it to a traveling salesman problem, and divide it into sub-linear programming problems without integer variables. A forward recursive algorithm with heuristic adjustments is proposed to solve it. When unit variable production costs are equal and goodwill loss rate is zero, the algorithm can obtain optimal solutions. Under other situations, numerical comparisons with CPLEX 12.6.2 show our algorithm can reach optimal in most cases and has computation time advantage for large-size problems. Numerical tests also demonstrate that initial capital availability as well as loan interest rate can substantially affect the manufacturer's optimal lot sizing decisions.
\end{abstract}

\begin{keyword}
lot sizing; customer good will; capital flow; profit maximization; loan
\end{keyword}

\end{frontmatter}

%\linenumbers

\section{Introduction}

The lot sizing problem was first introduced and solved by \citet{wagner1958dynamic1}. They proposed a polynomial algorithm to solve the single-item uncapacitated deterministic lot sizing problem, which has a computational complexity of $O(T^{2})$ and $T$ is the length of planning horizon. \citet{wagelmans1992economic} developed an $O(T\lg T)$ algorithm for the Wagner-Whitin cases. There is now abundant literature in this area that extends the basic model, such as the capacitated lot sizing problem, multi-item lot sizing problem, multi-level lot sizing problem, stochastic lot sizing problem, etc. This has also resulted in the inflation of the problems complexity. Mathematical programming heuristics, Lagrangian relaxation heuristics, decomposition and aggregation heuristics, meta heuristics, problem-specific greedy heuristics, piecewise linear approximation methods are used to solve different lot sizing problems. Some works adopting those methods can be found in \cite{gonzalez2011heuristic10}, \cite{absi2013heuristics20}, \cite{toledo2015relax}, \cite{rossi2015piecewise}, \cite{molina2016mip}. Comprehensive reviews on lot sizing problem could be addressed in \cite{maes1988multi2}, \cite{karimi2003capacitated3}, \cite{brahimi2006single4}, \cite{jans2007meta5}, \cite{buschkuhl2010dynamic6}, \cite{brahimi2017single}. 

In addition to large amount of papers setting the objective to minimize total cost in the lot sizing problems, there are also some works that formulate profit maximization models. \citet{aksen2003single7} developed a forward recursive dynamic programming algorithm to solve a single-item lot sizing problem with immediate lost sales for a profit maximization model. \citet{berk2008single8} investigated the single-item lot sizing problem for a warm/cold process with immediate lost sales and established theoretical results on the structure of optimal solutions.  \citet{haugen2007profit9} built a multi-product capacitated lot sizing profit maximization model in which there is a negative relation between product price and customer demands. \citet{gonzalez2011heuristic10} also addressed a multi-product capacitated lot sizing problem with pricing, setup time, and more general holding costs. \citet{sereshti2013profit11} studied the scheduling problem with demand choice flexibility and evaluated the efficiency of two mathematical models. However, those works didn't take capital flow constraints in their profit maximization models.

Loss of goodwill has been taken into consideration by some researchers in many inventory models. \citet{shi2004integrated12} employed goodwill loss in a news vendor model. \cite{pentico2009deterministic113,pentico2009deterministic214} considered loss of goodwill in a deterministic EOQ model and EPQ model respectively. \citet{chen2009production15} assumed supply shortage and time window violation will cause goodwill loss in a production scheduling and vehicle routing problem. In the lot sizing model, the concept of goodwill loss was first introduced by  \citet{hsu2001dynamic16}. In the above papers, customer goodwill loss was assumed to result in a penalty cost. Some researchers began to consider that the loss of customer goodwill would manifest itself in terms of reduced future sales. This was first observed by \citet{graves1999manufacturing17}. Empirical evidence presented by \citet{blazenko2003corporate18} showed that goodwill loss occurred when there were alternative sources of supply for customers, and if there were better alternative sources, the prospect of long-term revenue loss was greater.  \citet{aksen2007loss19} developed a single-item lot sizing model in which the unsatisfied demand in a given period caused the demand in the next period to shrink due to the loss of customer goodwill. \citet{absi2013heuristics20} addressed the multi-item capacitated lot sizing problem with setup times and lost sales, and they used a Lagrangian relaxation of the capacity constraints to decompose it into single-item uncapacitated sub-problems.

In business transactions, once a manufacturer encounters capital shortage, it needs to borrow money to maintain production; otherwise, it has to reduce or even cancel the production and could not provide sufficient products for their clients. A survey of 531 businesses that went bankrupt during the calendar year 1998 in \citet{bradley2000lack} pointed out that inadequate financial planning was one of main reasons for their business failing. A report by \citet{coughtrie2009restructuring} showed shortage of capital accounted for 17\% of company bankruptcies in Australia in 2008. \citet{elston2011financing22} found that 84\% of high-tech entrepreneurs in the US had experienced a shortage of capital at some time. To deal with capital shortage problem, loan is a widespread and effective option for many companies. After an agreement between a borrower and a lender is made, the borrower receives money from the lender, and is obligated to pay back an equal amount of money with the addition of some interests to the lender at a later time. A survey in \citet{Survey2014} about Small and medium enterprises (SMEs) in the 28 countries of the European Union showed that SMEs preferred to use bank loan, bank overdraft and trade credit. A report by \citet{IpsosMORI2017} based on the surveys of over 1000 SMEs from 2014 to 2016 in UK revealed that bank loan, friend loan and third party loan together accounted for about 40\% --- and ranked first --- of all external finances in the three years.

% A simple model in Uhrig-Homburg \cite{uhrig2005cash21} showed that cash flow shortage is an endogenous bankruptcy reason for the leveraged companies. Elston and Audretsch \cite{elston2011financing22} cited a survey data suggesting that 84\% of high-tech entrepreneurs in the US have experienced a shortage of capital at some time, with another 48\% currently experiencing liquidity constraints. An investigating report on the small-medium companies in Shen Zhen, China, showed that approximately 75\% percent of small-medium companies are experiencing capital shortage, and that the shortage amount is up to 1 billion dollars \cite{report201223}.

Relevant works taking capital flow or financing into account in inventory management problems are the following. \citet{buzacott2004inventory} adopted a news vendor model to analyze the importance of joint consideration of production and financial decisions.  \citet{chao2008dynamic} investigated a multi period news vendor problem constrained by cash flow and proved the optimality of a base stock policy.  \citet{gong2014dynamic} extended the model by considering short term financing.  \citet{zeballos2013single} built a periodic review inventory problem with working capital constraints, payment delay and multiple sources of financing. The above mentioned works are not for lot sizing problem and there are no fixed ordering costs in their models. Considering capital flow constraints, \citet{Chen161} built a single-item lot sizing model with trade credit and apply a dynamic programming algorithm to solve it.

From the literature review above we can find that, most previous works on lot sizing problems seldom consider the influence of capital flow constraints and external financing to production planning. This, together with our discussion on the importance of capital management and widespread use of loan by manufacturers, motivate our study. The main contributions of this paper are the following. 
\begin{itemize}
\item We introduce capital flow constraints to the traditional lot sizing problems and formulate a profit maximization model.
\item Optimality structures of the solution to the problem is discussed, and we develop a polynomial forward recursive algorithm with some heuristic adjustments to solve the problem. 
\item A common supply chain financing behavior, loan, is also introduced and discussed in the lot sizing problem.
\end{itemize}

The rest of this paper is organized as follows. Section 2 formulates the mathematical model and discusses its NP-hardness. Section 3 gives some mathematical properties and approximates this model to a traveling salesman problem. Section 4 divides the problem into sub-linear problems and proposes an algorithm with some heuristic adjustment techniques. Section 5 implements the numerical analysis: use some numerical cases to show the influence of capital flow constraints; compare the performance of our algorithm on large random generated test cases with CPLEX; analyze the main factors affecting the performance of our algorithm. Finally, section 6 concludes the paper and outlines future research directions.

\section{Problem description}

\subsection{Notations and assumptions}
We adopt the following notations for our model. Some relevant notations will be introduced when needed. 

\noindent
\begin{tabular}{ll}
$t\quad$     &index of a period, $t=1,2,\dots,T$.\rule{0 pt}{5 mm}\\%控制这一行高度
$d_{t}\quad$ &demand in period $t$.\\
$p_{t}\quad$ &unit selling price in period $t$.\\
$c_{t}\quad$ &unit production cost (variable cost) in period $t$.\\
$h_{t}\quad$ &unit inventory holding cost in period $t$.\\
$s_{t}\quad$ &production launching cost (fixed cost) in period $t$.\\

$B_{c}\quad$ &quantity of self-owned capital at the beginning of period 1.\\
$B_{L}\quad$ &quantity of loan at the beginning of period 1.\\
$B_{0}\quad$ &total initial capital at the beginning of period 1 and $B_{0}=B_{c}+B_{L}$.\\
$I_{0}\quad$ &initial inventory level at the beginning of period 1.\\
$T_{L}\quad$ &length of loan, $T_{L}\leq T$.\\
$r\quad$     &interest rate of loan.\\
$\beta\quad$ &customer goodwill loss rate.\\
$M\quad$     &a large number.
\end{tabular}

\vspace{3mm}
The decision variables used in the models include the following:\\
\begin{tabular}{ll}
$B_{t}\quad$ &end-of-period capital for period $t$.\rule{0 pt}{5 mm}\\
$I_{t}\quad$ &end-of-period inventory level for period $t$.\\
$x_{t}\quad$ &{a binary variable signaling whether the production occurs in period $t$.}\\
$y_{t}\quad$ &production quantity in period $t$.\\
$w_{t}\quad$  &demand shortage (lost sales) in period $t$, and we define $w_{0}=0$.\\
$Ed_{t}\quad$ &effective demand in period $t$ when considering customer goodwill loss.\\
$v_{t}\quad$  &realized demand in period $t$ and $ v_{t}= Ed_{t}-w_{t}$.\\
$\delta_{t}\quad$  &a binary variable signaling whether $Ed_{t}$ is positive.
\end{tabular}

\vspace{3mm}
In our problem, we make the following assumptions:
\begin{enumerate}[partopsep=1pt,itemsep=0pt,parsep=0pt]%,leftmargin=*
  \item Initial  capital of period $t$ should not be less than the total production cost in period $t$, namely, $B_{t-1}\geq s_{t}x_{t}+c_{t}y_{t}$, in which initial capital of period $t$ is $B_{t-1}$, and the total production cost in period $t$ is production launching cost $s_{t}x_{t}$ plus variable production cost $c_{t}y_{t}$.\label{assume1}
  \item End-of-period capital for each period should be not less than 0, namely, $B_{t}\geq 0$.\label{assume2}
  \item Initial inventory of the planning horizon is 0, namely, $I_{0}=0$.\label{assume3}
  \item No backorder is allowed.\label{assume4}
  \item The manufacturer could decide the realized quantity for customer's demand without paying penalty cost, but lost sales can cause the shrinking of demand in the next period.\label{assume5}
  \item The manufacturer uses loan in the first period and pays back the principal and interest after a certain length of periods; length of loan is shorter than the length of total planning horizon; mortgages are not required for loan. \label{assume6}
\end{enumerate}

The biggest difference between our problem and traditional lot sizing problem is Assumption \ref{assume1} and \ref{assume2}: how much to produce is constrained by present capital, and end-of-period capital for each period should be above zero to avoid bankruptcy. Assumption \ref{assume3} and \ref{assume4} are also the standard assumptions of the Wagner-Whitin model. 
Assumption \ref{assume5} means the manufacturer can decide how much products they want to provide for customers. Assumption \ref{assume6} defines the loan type in our paper. Loan length is smaller than the total planning horizon because lenders tend to have higher risks for longer loan length and they generally don't provide loan out of the planning horizon in order to reduce risks.

\subsection{Mathematical models for our problem}

Considering loss of customer goodwill, effective demand is the remnant demand after the goodwill loss from the original demand. As in \citet{aksen2007loss19}, the effective demand in period $t$ can be represented as:
\begin{align}
Ed_{t}=\max\{0,d_{t}-\beta w_{t-1}\},\quad   \forall t. \label{eq:edt}
\end{align}

Eq. \eqref{eq:edt} is a nonlinear equation. For convenience of computation, we bring several linear constraints to replace this nonlinear equation and construct the mixed integer programming model below.
\\

\textbf{Model P}
\begin{alignat}{2}
 &\max\quad && B_{T}-B_{c}-B_{L}=\sum_{t=1}^{T}\left[p_{t}(Ed_{t}-w_{t})-(h_{t}I_{t}+s_{t}x_{t}+c_{t}y_{t})\right]-B_{L} (1+r)^{T_{L}}\label{eq:objective}\\
&s.t.             &&\text{for }t=1,2,\dots,T \nonumber\\
&                        && y_{t}\leq Mx_{t},\label{con:setuppositive}\\ 
&                        &&s_{t}x_{t}+c_{t}y_{t}\leq B_{t-1},\label{con:capital}\\
&                        &&w_{t}\leq Ed_{t},\label{con:wtedt}\\  
&                        &&I_{t}=I_{t-1}+y_{t}-Ed_{t}+w_{t}, \label{eq:inventoryflow} \\
&                        && B_{0}=B_{c}+B_{L}\label{eq:inicapital},\\
&                        &&B_{t}=
\begin{cases} 
 B_{t-1}+p_{t}(Ed_{t}-w_{t})-h_{t}I_{t}-s_{t}x_{t}-c_{t}y_{t},\quad &t\neq T_{L},\\
 B_{t-1}+p_{t}(Ed_{t}-w_{t})-h_{t}I_{t}-s_{t}x_{t}-c_{t}y_{t}-B_{L}(1+r)^{T_{L}},\quad &t=T_{L},\\
 \end{cases}\label{eq:capitalflow}\\
 &                       &&d_{t}\leq \beta w_{t-1}+M\delta_{t}, \label{con:dtwt1}\\
&                        &&d_{t}\geq \beta w_{t-1}-M(1-\delta_{t}), \label{con:dtwt2}\\
&                        &&Ed_{t}\leq d_{t}-\beta w_{t-1}+M(1-\delta_{t}), \label{con:edtdt1}\\
&                       &&Ed_{t}\geq d_{t}-\beta w_{t-1}-M(1-\delta_{t}), \label{con:edtdt2}\\
&                       &&Ed_{t}\leq d_{t}\delta_{t}, \label{con:edtdt3}\\
&                      &&I_{0}=0,I_{t}\geq 0,\label{eq:i0it}\\  
&                      &&Ed_{t}\geq 0, w_{t}\geq 0,y_{t}\geq 0,\label{con:edtwtyt}\\  
&                      &&x_{t}\in \{0,1\}, \delta_{t}\in\{0,1\}. \label{con:xtdeltat} 
\end{alignat}

The objective defined by Eq. \eqref{eq:objective} is to maximize the capital increment from the beginning of the planning horizon to the final period, where the realized sales in period $t$ are given by $Ed_{t}-w_{t}$, revenue in period $t$ is $p_{t}(Ed_{t}-w_{t})$, total cost in period $t$ is $h_{t}I_{t}+s_{t}x_{t}+c_{t}y_{t}$, or $h_{t}I_{t}+s_{t}x_{t}+c_{t}y_{t}+B_{L}(1+r)^{T_{L}}$ if period $t$ should pay back the loan.  

Constraint \eqref{con:setuppositive} enforces setups with positive production in each period. Constraint \eqref{con:capital} represents Assumption \ref{assume1} and \ref{assume2}:  initial  capital in period $t$ should be not less than the total production cost in period $t$. It also ensures the non-negativity of $B_{t-1}$, which avoids bankruptcy. Constraint \eqref{con:wtedt} ensures that any lost demand $w_{t}$ in period $t$ not exceed the effective demand of that period. Constraint \eqref{eq:inventoryflow} provides the inventory flow balance equation, and Constraints \eqref{eq:inicapital} and \eqref{eq:capitalflow} define the capital flow balance.

Constraints \eqref{con:dtwt1}-\eqref{con:edtdt3} are the linear descriptions of Eq. \eqref{eq:edt}. Among them, Constraints \eqref{con:dtwt1} and \eqref{con:dtwt2} ensure the binarity of $\delta_{t}$: if $d_{t}\leq \beta w_{t-1}$, $\delta_{t}=0$ and effective demand $Ed_t$ is also 0, else, $\delta_{t}=1$ and effective demand $Ed_t$ is positive;  Constraint \eqref{con:edtdt1}, \eqref{con:edtdt2} and \eqref{con:edtdt3} determine the value of effective demand $Ed_{t}$: if $d_{t}\leq\beta w_{t-1}$, $\delta_{t}=0$ and $Ed_{t}=0$; else, $\delta_{t}=1$ and $Ed_{t}=d_{t}-\beta w_{t-1}$. Constraints \eqref{eq:i0it}, \eqref{con:edtwtyt} and \eqref{con:xtdeltat} guarantee the non-negativity and binarity of variables. Constraint \eqref{eq:i0it} also represents Assumption \ref{assume3} and \ref{assume4} in our problem.

Note that if $\beta=0$, this model is transformed to a capital flow constrained problem without loss of goodwill.

\subsection{Computational complexity of Model P}
The single-item capacitated lot sizing problem has been shown by \citet{bitran1982computational} to be NP-hard. As for our problem, the capital flow constraint $c_{t}y_{t}+s_{t}x_{t}\leq B_{t-1}$ is a capacity constraint by removing $s_{t}x_{t}$ and replacing $B_{t-1}$ with $C_{t}$, where $C_{t}$ is the capacity in period $t$. Therefore, capital flow constraint is a special type of capacity constraints, and model P is also a NP-hard problem. 

\section{Mathematical properties}
\label{Sec3}
To describe the mathematical properties of our problem, we first define the concepts of \emph{production cycle} and \emph{production round}.

\begin{definition}
In a production plan, if the manufacturer launches production at the beginning  of period $m$, and it does not launch new production till the end of period $t$ ($m\leq t\leq T$), we call period $m$ to period $t$ a production cycle. 
\end{definition}

\begin{definition}
In a production plan, for one or more consecutive production cycles that last from the beginning of $m$ to the end of period $t$, if initial inventory for period $m$ and end-of-period inventory for period $t$ are zeros, we call period $m$ to period $t$ a production round. 
\end{definition}

\subsection{Properties about initial inventory and capital}
\begin{lemma} 
If unit variable production costs are equal, namely, $c_{t}=c$, $\forall t$, for any production cycle starting at period $t+1$ ($t+1=1,\dots,T$),  the optimal solution satisfies $I_{t} y_{t+1}=0$. \label{lemma1}
\end{lemma}

\begin{proof}
Lemma 1 is apparently true when $t+1=1$ because $I_{0}=0$. For $t+1>1$, if there is a solution that does not satisfy Lemma 2, namely, $I_{t}>0$ and $y_{t+1}>0$, assume that period $t+1$'s former production cycle begins at period $m$ ($1\leq m\leq t$), and the production cycle beginning at period $t+1$ lasts till period $n$ ($n\leq T$). The production plan is shown in Figure \ref{fig 2}.

\begin{figure}[h]  %插入图形时用latex编译
\centering\includegraphics[scale=0.8]{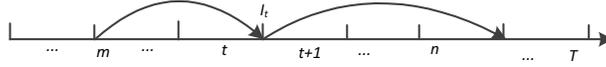}
\caption{Sketch of production plan when $I_{t}>0$ and $y_{t+1}>0$.}\label{fig 2}
\end{figure}

According to capital flow balance equation \eqref{eq:capitalflow}, end-of-period capital for period $t$ and period $T$ are given by the following equations.
\begin{align}
B_{t}&=\begin{cases}
B_{m-1}+\sum\nolimits_{i=m}^{t}p_{i}(Ed_{i}-w_{i})-s_{m}-c_{m}y_{m}-\sum\nolimits_{i=m}^{t}h_{i}I_{i},\quad &t\neq T_{L},\\
B_{m-1}+\sum\nolimits_{i=m}^{t}p_{i}(Ed_{i}-w_{i})-s_{m}-c_{m}y_{m}-\sum\nolimits_{i=m}^{t}h_{i}I_{i}-B_{L}(1+r)^{T_{L}}, &t=T_{L},
\end{cases}\\
B_{T}&=
\begin{cases}
B_{t}+\sum\nolimits_{i=t+1}^{T}\left[p_{i}(Ed_{i}-w_{i})-(h_{i}I_{i}+s_{i}y_{i}+c_{i}y_{i})\right],\quad &T\neq T_{L},\\
B_{t}+\sum\nolimits_{i=t+1}^{T}\left[p_{i}(Ed_{i}-w_{i})-(h_{i}I_{i}+s_{i}y_{i}+c_{i}y_{i})\right]-B_{L}(1+r)^{T_{L}}, &T=T_{L}.
\end{cases}
\end{align}

If the production quantity in period $m$ reduces $I_{t}$ and the production quantity in period $t+1$ increases $I_{t}$, then the effective demands $Ed_{i}$ ($i=1,\dots,T$) would not be influenced and the production plan is still feasible. $B_{t}$ and $B_{T}$ change to the following:
\begin{align}
B_{t}^{'}&=B_{t}+c_{m}I_{t}+\sum\nolimits_{i=m}^{t}h_{i}I_{t},\\
B_{T}^{'}&=B_{T}+c_{m}I_{t}+\sum\nolimits_{i=m}^{t}h_{i}I_{t}-c_{t+1}I_{t}.
\end{align}

If unit variable production costs are equal, $B_{T}^{'}=B_{T}+\sum\nolimits_{i=m}^{t}h_{i}I_{t}\geq B_{T}$, The final  capital increases. Therefore, $I_{t}>0$ and $y_{t+1}>0$ is not the optimal solution; optimal solution always satisfies $I_{t}y_{t+1}=0$.
\end{proof}

Lemma \ref{lemma1} is also known as the \emph{zero-inventory-ordering policy}, which means initial inventory of a production cycle is always zero. Define $f_t(I_{t-1},B_{t-1},w_{t-1})$ as the maximum capital increment during period $t, t+1,\dots,T$, given period $t$'s initial inventory $I_{t-1}$, initial capital $B_{t-1}$, and previous period's demand shortage quantity $w_{t-1}$. 

\begin{lemma} 
For any period $t$ ($t=1,\dots,T$), when initial inventory $I_{t-1}$ and last period's demand shortage $w_{t-1}$ are fixed, $f_t(I_{t-1},B_{t-1},w_{t-1})$ is nondecreasing with period $t$'s initial capital $B_{t-1}$. \label{lemma2}
\end{lemma}

\begin{proof}
To prove this property, we build a dynamic programming model for our problem. For any at period $t$ ($t=1,\dots,T$), its states are: initial inventory $I_{t-1}$, initial  capital is $B_{t-1}$, and demand shortage quantity of previous period $w_{t-1}$. Its actions are production quantity $y_{t}$ and demand realized quantity $v_{t}$. Lower bounds for $y_{t}$ and $v_{t}$ are both 0; Upper bound for $y_{t}$ is the maximum production quantity under present capital and upper bound for  $v_{t}$ is the quantity of effective demand in this period, which are shown by Eq. \eqref{eq:Xupbound} and Eq. \eqref{eq:vupbound}, respectively.
\begin{align}
\overline{y}_{t}=&\max\big\{0,(B_{t-1}-s_{t})/c_{t}\big\},\label{eq:Xupbound}\\
\overline{v}_{t}=&\max\{0,\beta(d_{t}-w_{t-1})\}.\label{eq:vupbound}
\end{align}

Define a unit step function $K(z)$: $K(z)=1$ if $z>0$;  $K(z)=0$ if $z\leq 0$. The state transition equations are:
\begin{align}
I_{t}=&I_{t-1}+y_{t}-v_{t},\label{eq:dyinventoryflow}\\
B_{t}=&
\begin{cases} 
 B_{t-1}+p_{t}v_{t}-h_{t}I_{t}-s_{t}K(y_{t}>0)-c_{t}y_{t},\quad &t\neq T_{L},\\
 B_{t-1}+p_{t}v_{t}-h_{t}I_{t}-s_{t}K(y_{t}>0)-c_{t}y_{t}-B_{L}(1+r)^{T_{L}},\quad &t=T_{L},
 \end{cases}\label{eq:dycapitalflow}\\
 w_{t}=&\max\{0,\beta(d_{t}-w_{t-1})\}-v_{t}.
\end{align}

The functional equation for dynamic programming is:
\begin{equation}
\begin{cases}
f_{t}(I_{t-1},B_{t-1},w_{t-1})=\max\limits_{0\leq y_{t}\leq \overline{y}_{t}, 0\leq v_{t}\leq \overline{v}_{t}}\Big\{B_{t}-B_{t-1}+f_{t+1}(I_{t},B_{t},w_{t})\Big\},\\
f_{T+1}(I_{t},B_{t},w_{t})=0.
\end{cases}\label{eq:recuredynamic}
\end{equation}

when $I_{t-1}$ and $w_{t-1}$ are fixed, if increasing $B_{t-1}$, from Eq. \eqref{eq:Xupbound} and Eq. \eqref{eq:vupbound}, the feasible domain for $v_{t}$ does not change, but the feasible domain for $y_{t}$ stays constant or expands. There always exists actions $y'_{t}$ and $v'_{t}$ that make $f'_{t}(I_{t-1},B'_{t-1},w_{t-1})$ not lower than $f_{t}(I_{t-1},B_{t-1},w_{t-1})$. Therefore, $f_t(I_{t-1},B_{t-1},w_{t-1})$ is nondecreasing with $B_{t-1}$ with fixed $I_{t-1}$ and $w_{t-1}$.
\end{proof}

\begin{lemma} 
For any period $t$ ($t=1,\dots,T$), if goodwill lost rate is 0, when $I_{t-1}$ are fixed, $f_t(I_{t-1},B_{t-1})$ is nondecreasing with $B_{t-1}$. \label{lemma3}
\end{lemma}
\begin{proof}
The proof is similar to Lemma \ref{lemma2}. If goodwill shortage rate is 0, demand shortage is not the states of the dynamic programming problem. The actions are still $y_{t}$ and $v_{t}$ and state transition equations are Eq.  \eqref{eq:dyinventoryflow} and Eq. \eqref{eq:dycapitalflow}. The functional equation changes to be:
\begin{equation}
\begin{cases}
f_{t}(I_{t-1},B_{t-1})=\max\limits_{0\leq y_{t}\leq \overline{y}_{t}, 0\leq v_{t}\leq \overline{v}_{t}}\Big\{B_{t}-B_{t-1}+f_{t+1}(I_{t},B_{t})\Big\},\\
f_{T+1}(I_{t},B_{t})=0.
\end{cases}\label{eq:recuredynamic2}
\end{equation}

When $I_{t-1}$ are fixed and increasing $B_{t-1}$, feasible domain for $v_{t}$ keeps unchanged while feasible domain for $y_{t}$ stays constant or expands. Therefore, $f_t(I_{t-1},B_{t-1})$ is nondecreasing with $B_{t-1}$.
\end{proof}

\subsection{Approximated traveling salesman problem}

% Although many lot sizing problem can be transformed to a fixed charge network flow problem like the work \cite{zangwill1969backlogging} and apply its properties, our problem is different in some aspects: production quantities are influenced by capital flow, while capital flow is influenced by many parameters; capital flow in the former periods have impact on the capital flow in latter periods. Therefore, some properties in the fixed charge network flow problem does not fit for our problem.

% However, 

% Based on Lemma 1, Lemma 2 and Lemma 3, we use the following heuristic steps to  convert Models P into a traveling salesman problem. 
% \begin{enumerate}[partopsep=0pt,itemsep=0pt,parsep=0pt]
% \item we apply the zero-inventory ordering policy to divide production plan into several production rounds.
% \item although end-of-period capital $B_{t}$ are influenced by its initial capital and previous demand shortage, we omit the impact of previous demand shortage, and always select maximum initial capital for computation; if initial capital is same, we select the one that gives minimum previous demand shortage. 
% Note that the first one is not a heuristic step if variable production costs are stationary.
% \end{enumerate}

Define $BB_{m,n}^{B_{m-1}}$ as the maximum capital increment in a production round from period $m$ to period $n$ with initial capital $B_{m-1}$. Based on Lemma \ref{lemma1}, Lemma \ref{lemma2} and Lemma \ref{lemma3}, our problem is approximately transformed to a traveling salesman problem finding the longest route as shown in Figure \ref{fig:tsp} (in this case, $T=4$). Main ideas behind this approximation are that we apply the zero-inventory-ordering policy to divide production plan into several production rounds, and apply Lemma \ref{lemma2} and Lemma \ref{lemma3} to always select maximum initial capital for computation of the length of arcs in the traveling salesman problem.

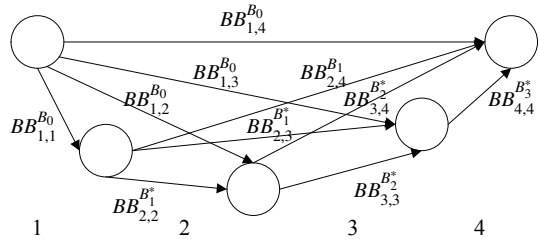
\begin{figure}[h]  
\centering\footnotesize
\begin{overpic}[scale=0.7]{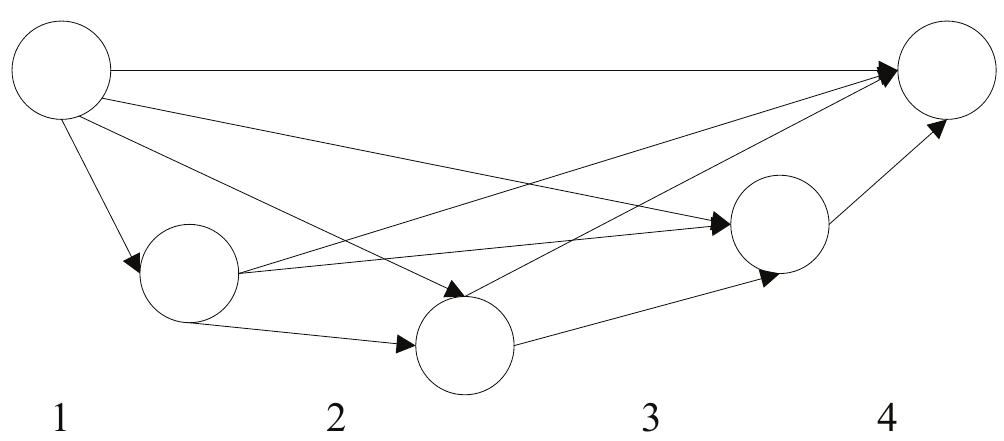}
\put(1,20){$BB_{1,1}^{B_{0}}$}
\put(22,25){$BB_{1,2}^{B_{0}}$}
\put(35,30){$BB_{1,3}^{B_{0}}$}
\put(40,40){$BB_{1,4}^{B_{0}}$}
\put(20,5){$BB_{2,2}^{B_{1}^{\ast}}$}
\put(45,20){$BB_{2,3}^{B_{1}^{\ast}}$}
\put(55,30){$BB_{2,4}^{B_{1}}$}
\put(65,8){$BB_{3,3}^{B_{2}^{\ast}}$}
\put(63,25){$BB_{3,4}^{B_{2}^{\ast}}$}
\put(90,25){$BB_{4,4}^{B_{3}^{\ast}}$}
\end{overpic}
\caption{Approximated traveling salesman problem.}\label{fig:tsp}
\end{figure}

A functional equation is constructed for computation of the problem, and we set $B_0^{\ast}=B_0$.
\begin{equation}
B_{n}^{\ast}=\max\limits_{1\leq m< n}\left[B_{m-1}^{\ast}+BB_{m,n}^{B_{m-1}^{\ast}}\right],\quad n=1,2,\dots T.\label{eq:recurecomputation}
\end{equation}

Apparently, $B_{T}^{\ast}=\max\limits_{1\leq m\leq T}\left[B_{m-1}^{\ast}+BB_{m,T}^{B_{m-1}^{\ast}}\right]$. On this functional equation, we have the following properties.

\begin{lemma}
If unit variable production costs are equal and goodwill loss rate is 0, for any period $t$, given end-of-period inventory of period $t$ is 0, the optimal production plan from period 1 to $t$ is part of the optimal production plan from period 1 to $T$.\label{lemma4}
\end{lemma}
\begin{proof}
If goodwill loss rate and end inventory level for period $t$ are both 0, based on Eq. \eqref{eq:recuredynamic2}, given initial capital and initial inventory, maximum capital increment during period $t$ to period $T$ is:
\begin{equation}
f_{t}(I_{t-1},B_{t-1})=B_{T}-B_{t-1}=\max\limits_{0\leq y_{t}\leq \overline{y}_{t}, 0\leq v_{t}\leq \overline{v}_{t}}\Big\{B_{t}-B_{t-1}+f_{t+1}(0,B_{t})\Big\}.\label{eq:recurezeroinventory}
\end{equation}

From Lemma \ref{lemma3}, $f_{t+1}(0,B_{t})$ is a non-decreasing function of $B_t$, and gets the maximum value when $B_{t}$ is the maximum end-of-period capital for period $t$. Given zero end-of-period inventory in period $t$, this happens when the production plan from period 1 to $t$ is optimal because optimal production plan satisfies zero-inventory-ordering policy when the variable production costs are equal. So, what is optimal for period 1 to period $t$ is also optimal for $f_{t+1}(0,B_{t})$, which is the maximum capital increment during period $t+1$ to period $T$.  Therefore, the optimal production plan from period 1 to $t$ is part of the optimal production plan from period 1 to $T$.
\end{proof}

\begin{lemma}
If unit variable production costs are equal and goodwill loss rate is 0, for any two period $t_1$, $t_2$ ($t_1<t_2$), given initial inventory of period $t_1$ is 0, end-of-period inventory of period $t_2$ is 0, the optimal production plan from period $t_1$ to $t_2$ is part of the optimal production plan from period 1 to $T$.\label{lemma5}
\end{lemma}
\begin{proof}
By lemma \ref{lemma4}, the optimal production plan from period 1 to $t_1-1$ is part of the whole production plan,  and the optimal production plan from period 1 to $t_1-1$ is also part of the optimal production plan from period 1 to $t_2$. If the production plan from period $t_1$ to $t_2$ is optimal, together with the optimal plan form period 1 to $t_1-1$, it will result in an optimal production plan from period $1$ to $t_2$ under the zero-inventory-ordering policy. Since the optimal production plan from period 1 to $t_2$ is part of the whole production plan, the optimal production plan from period $t_1$ to $t_2$ is part of the optimal production plan from period 1 to $T$.
\end{proof}

\begin{theorem}
If unit variable production costs are equal and goodwill lost rate is 0, namely, $c_{t}=c$, $\forall t$, and $\beta=0$, $B_{T}^{\ast}=\max\limits_{1\leq m\leq T}\left[B_{m-1}^{\ast}+BB_{m,T}\right]$ provides an optimal solution.\label{theorem1}
\end{theorem}

\begin{proof}
When the variable production costs are equal, Lemma \ref{lemma1} shows the problem satisfies the zero-inventory ordering policy. Hence, the optimal production plan can be a combination of several production rounds. The functional equation \eqref{eq:recurecomputation} in fact enumerates all the possible production rounds.  Lemma \ref{lemma4} and Lemma \ref{lemma5} indicate that the optimal production plan in a given production round is part of total optimal production plan. These are the same properties as the Wagner-Whitin case \citep{wagner1958dynamic1}. Therefore, for all the combinations of production rounds in the computation of $B_{T}^{\ast}=\max\limits_{1\leq m\leq T}\left[B_{m-1}^{\ast}+BB_{m,T}\right]$, the one that gives maximum final capital is the optimal solution. 
\end{proof}

When the variable production costs are not all equal, or goodwill loss rate is not 0, the functional equation in Theorem \ref{theorem1} only gets an approximate solution. However, based on some properties below, we devise some heuristic adjustments to make it close to the optimal solution.

\subsection{Properties for production plan adjustment}

\begin{corollary}
In a feasible solution, assume the solution is $x_{t}$, $y_{t}$, $w_{t}$ ($t=1,2,\dots ,T$). For any two consecutive production cycles, assume the former production cycle begins at period $t_{1}$, ends at $t_{2}-1$; the latter production cycle begins at period $t_{2}$, ends at period $t_{3}$, with end-of-period inventory at period $t_{3}$ is 0, end-of-period capital at period $t_{3}$ is $B_{t_{3}}$ and demand shortage at period $t_{3}$ is $w_{t_{3}}$. We can make a production plan adjustment from $t^{\prime}$ ($t_{1}+1\leq t^{\prime}\leq t_{2}-1$) to $t_{2}-1$ if the  capital $B_{t_{3}}^{\prime}\geq B_{t_{3}}$, shortage quantity $w_{t_{3}}^{\prime}\leq w_{t_{3}}$ and still with 0 end-of-period inventory in $t_{3}$ after this adjustment.\label{corollaryAdjust}
\end{corollary}

\begin{proof}
The feasible solution is presented by Figure \ref{fig:corollary1-1}. For period $t_{3}+1$, its functional equation is
\begin{equation}
f_{t_{3}+1}(I_{t_{3}},B_{t_{3}},w_{t_{3}})=\max\limits_{0\leq y_{t_{3}+1}\leq \overline{y}_{t_{3}+1}, 0\leq v_{t_{3}+1}\leq \overline{v}_{t_{3}+1}}\Big\{B_{t_{3}+1}-B_{t_{3}}+f_{t_{3}+2}(I_{t_{3}+1},B_{t_{3}+1},w_{t_{3}+1})\Big\}.\\
\end{equation}

From Lemma \ref{lemma2}, $f_{t_{3}+1}(I_{t_{3}},B_{t_{3}},w_{t_{3}})$ is non decreasing with $B_{t_{3}}$ when $I_{t_{3}}$ and $w_{t_{3}}$ are fixed. In the adjustment, $I_{t_{3}}$ are fixed to be 0, after increasing $B_{t_{3}}$ and decreasing $w_{t_{3}}$, the domains for $y_{t_{3}+1}$ and $v_{t_{3}+1}$ both expand from Eq. \eqref{eq:Xupbound} and Eq. \eqref{eq:vupbound}. Therefore, the final capital increment is non decreasing with this adjustment. 
\end{proof}

The adjustment in Corollary \ref{corollaryAdjust} is shown by Figure \ref{fig:corollary1-2}.

\begin{figure}[ht]
\centering\includegraphics[scale=0.8]{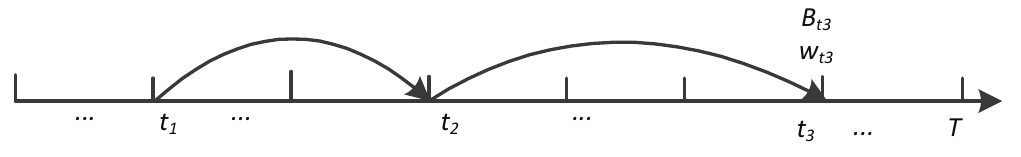}
\caption{A feasible solution.}\label{fig:corollary1-1}
\centering\includegraphics[scale=0.8]{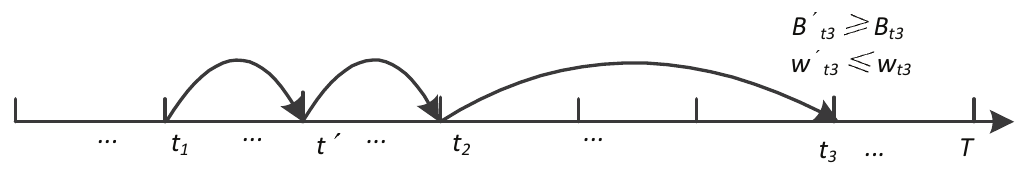}
\caption{Heuristic adjustment based on Corollary 1.}\label{fig:corollary1-2}
\end{figure}

\begin{corollary}
In a feasible solution, assume that the solution is $x_{t}$, $y_{t}$, $w_{t}$ ($t=1,2,\dots ,T$). For any two consecutive production cycles, assume that the former production cycle begins at period $t_{1}$, ends at period $t_{2}-1$, and the latter production cycle begins at period $t_{2}$. If $B_{t_{1}-1}-s_{t_{1}}-c_{t_{1}}y_{t_{1}}>0$, $c_{t_{1}}+\sum_{i=t_{1}}^{t_{2}-1}h_{i}<c_{t_{2}}$, then it is better to move some production amount $\Delta y_{t_{2}}$ from $y_{t_{2}}$ to $y_{t_{1}}$ to obtain more final capital.\label{corollaryInventory}
\end{corollary}

\begin{proof}
This heuristic step is shown in Figure \ref{fig:movingamount}. If $B_{t_{1}-1}-s_{t_{1}}-c_{t_{1}}y_{t_{1}}>0$, production cycle $t_{1}$ has residual production capacity, which could produce more. After the moving adjustment, the final  capital changes to the following:
\begin{align}
B_{T}^{'}=B_{T}+\left(c_{t_{2}}-c_{t_{1}}-\sum\nolimits_{i=t_{1}}^{t_{2}-1}h_{i}\right)\Delta y_{t_{2}}.
\end{align}

If $c_{t_{1}}+\sum_{i=t_{1}}^{t_{2}-1}h_{i}<c_{t_{2}}$, this adjustment does not affect the feasibility of the solution and $B_{T}^{'}>B_{T}$. Therefore, final capital increases.
\end{proof}

The moving production amount $\Delta y_{t_{2}}$ can be obtained by Eq. \eqref{eq:movingamount}.

\begin{equation}
\Delta y_{t_{2}}=\frac{{B_{t_{1}-1}-s_{t_{1}}}}{{c_{t_{1}}}}-y_{t_{1}}.\label{eq:movingamount}
\end{equation}
%\min\left\{\right\},\frac{B_{t_{2}-1}-s_{t_{2}}-c_{t_{2}}y_{t_{2}}}{c_{t_{2}}-c_{t_{1}}-\sum_{i=t_{1}}^{t_{2}-1}h_{i}}

Eq. \eqref{eq:movingamount} is the maximum production quantity increment that cycle $t_{1}$ can provide. 

\begin{figure}[h]
\centering\includegraphics[scale=0.8]{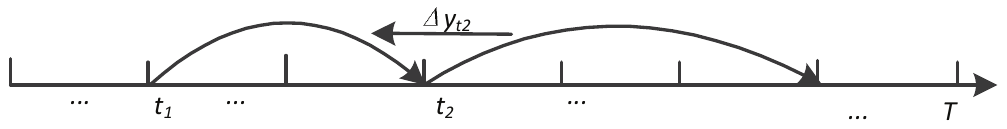}
\caption{Heuristic adjustment based on Corollary \ref{corollaryInventory}.}\label{fig:movingamount}
\end{figure}

\section{Sub-problems and algorithm for our problem}

For the computation of $BB_{m,n}$ in the recursive equation \eqref{eq:recurecomputation}, we remove the integers of Model P and divide it into sub-linear problems. In the sub-linear variables, only realized demand $v_t$, $\forall t$, are decision variables. We also devise some heuristic techniques to adjust the production plan.

\subsection{sub-linear problems}
\label{sec3.1}

By definition, $BB_{m,n}$ is the maximum  capital increment in a production round. Therefore, it may include several production cycles. For a production round with fixed $k$ production cycles, assume the production launching periods are $t_{1}$, $t_{2}$, $\dots$, $t_{k}$ (for convenience of expression, we set $m=t_{1}$, $n=t_{k+1}-1$), the production plan is shown in Figure \ref{fig:BBmn}.

\begin{figure}[!ht]  
\centering\includegraphics{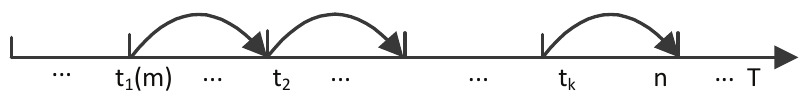}
\caption{Meaning of $BB_{m,n}$.}\label{fig:BBmn}
\end{figure}

To compute $BB_{m,n}$, there are still integer variables $\delta_{t}$, which is a 0-1 variable indicating whether previous goodwill loss is below effective demand. In Model P-sub1 below, we use a heuristic step by assuming $\delta_t=1$, $m\leq t\leq n$, namely, all demand from period $m$ to period $n$ are above previous goodwill loss. With known $w_{t_{1}-1}$, we convert Model P to a sub linear problems P-sub1.
\\

\textbf{Model P-sub1}
\begin{alignat}{2}
 &\max\quad && BB_{m,n} =\max\{B_{n}-B_{m-1}\}\label{eq:BBmn} \\
&s.t.      && \text{for }t=m,m+1,\dots ,n  \nonumber\\
&                 && c_{t_{i}}\sum_{j=t_{i}}^{t_{j+1}-1}v_{j}+s_{t_{i}}\leq B_{t_{i}-1},\quad i=1,2,\dots,k,\label{con1:setupcapital}\\
&                 &&B_{t}\geq 0,\label{con1:Babove0}\\
&                 &&I_{t}=\sum_{j=t+1}^{t_{i+1}-1}v_{j},\quad t_i<t<t_{i+1},i=1,2,\dots,k,\label{con1:inventoryflow}\\
&                 &&B_{t}=
\begin{cases}
B_{t-1}+p_{t}v_{t}-(h_{t}I_{t}+s_{t}+c_{t}\sum_{j=t_{i}}^{t_{i+1}-1}v_{j}),\quad & t\neq T_{L},t=t_{i},i=1,2,\dots,k,\\
B_{t-1}+p_{t}v_{t}-h_{t}I_{t},\quad & t\neq T_{L},t\neq t_{i},i=1,2,\dots,k,\\
\end{cases}\label{con1:capitalflow1}\\
&                 &&B_{t}=
\begin{cases}
B_{t-1}+p_{t}v_{t}-(h_{t}I_{t}+s_{t}+c_{t}\sum_{j=t_{i}}^{t_{i+1}-1}v_{j})-T_{L}(1+L)^{L},\quad & t=T_{L},t=t_{i},i=1,2,\dots,k,\\
B_{t-1}+p_{t}v_{t}-h_{t}I_{t}-T_{L}(1+L)^{L},\quad & t= T_{L},t\neq t_{i},i=1,2,\dots,k,\\
\end{cases}\label{con1:capitalflow2}
\\
&                 &&I_{t_{1}-1}=0, I_{t_{2}-1}=0, \dots, I_{t_{k}-1}=0, I_{n}=0,\label{con1:zeroinventory}\\
&                 && Ed_{t}=
\max\{0,d_{t}-\beta w_{t-1}\},\quad t=t_{1},\label{con1:ed0}\\
& &&
 Ed_{t}=d_{t}-\beta (Ed_{t-1}-v_{t-1}),\quad t\neq t_{1},\label{con1:edflow}\\
&                 &&0\leq v_{t}\leq Ed_{t}.\label{con1:uplowbounds}
\end{alignat}

The objective function \eqref{eq:BBmn} is to maximize capital increment from a production round starting at period $m$ and ending at period $n$. Constraints \eqref{con1:setupcapital} and \eqref{con1:Babove0} represents our paper's assumptions 1 and 2 about capital flow constraints. Constraint \eqref{con1:inventoryflow} shows the relationship between $I_{t}$ and $v_{t}$. Constraint \eqref{con1:capitalflow1} and Constraint \eqref{con1:capitalflow2} are the capital flow balance. Constraint \eqref{con1:zeroinventory} means that initial inventory and end-of-period inventory of each production cycle are both zeros, which is a heuristic step if unit variable production costs are not equal. Constraints \eqref{con1:ed0} and \eqref{con1:edflow} are expressions of effective demands. Constraint \eqref{con1:edflow} also reflects the heuristic assumption:  $\delta_{t}=1$, $m\leq t\leq n$.
Constraint \eqref{con1:uplowbounds} provides the lower and upper bounds of variables $v_t$, which is realized demand in period $t$.

If Model P-sub1 does not obtain a feasible solution, this may be related with the heuristic assumption of $\delta_{t}$. So next step we relax this assumption without loss of goodwill: removing Constraint \eqref{con1:edflow} in Model P-sub1, amending  Constraint \eqref{con1:uplowbounds} to get another sub linear problem P-sub2 below.
\\

\textbf{Model P-sub2}
\begin{alignat}{2}
 &\max\quad && BB_{m,n} =\max\{B_{n}-B_{m-1}\}\tag{\ref{eq:BBmn}}\\
&s.t.      && \text{for }t=m,m+1,\dots ,n  \nonumber\\
&                 &&\eqref{con1:setupcapital}\text{--}\eqref{con1:ed0}\\ 
&                 &&0\leq v_{t}\leq d_{t}.
\end{alignat}
 
Based on the solution of Model P-sub2, compute $\delta_{t}$ according to equations \eqref{eq:wtvt} and \eqref{eq:deltat} below.
\begin{align}
w_{t}&=Ed_{t}-v_{t},\quad t=m,m+1\dots,n,\label{eq:wtvt}\\
\delta_{t}&=
\begin{cases}
0,\quad & \text{if } d_{t}-\beta w_{t-1}<0,t=m,m+1\dots,n, \\
1,\quad & \text{if } d_{t}-\beta w_{t-1}\geq 0, t=m,m+1\dots,n.
\end{cases}\label{eq:deltat}
\end{align}

Based on the values of $\delta_{t}$, another sub linear problem is formulated.
\\

\textbf{Model P-sub3}
\begin{alignat}{2}
 &\max\quad && BB_{m,n} =\max\{B_{n}-B_{m-1}\}\tag{\ref{eq:BBmn}}\\
&s.t.  && \text{for }t=m,m+1,\dots ,n  \nonumber\\
&                 &&\eqref{con1:setupcapital}\text{--}\eqref{con1:ed0}\\ 
&              && Ed_{t}=\begin{cases}
d_{t}-\beta (Ed_{t-1}-v_{t-1}),\qquad &\text{if }\delta_{t}=1,\\
0,\qquad &\text{if }\delta_{t}=0,
\end{cases}\\
&              && d_{t}-\beta (Ed_{t-1}-v_{t-1})<0,\qquad \text{if }\delta_{t}=0.
\end{alignat}

% During numerical experiments, we find sometimes it is better to try to satisfy the demands in the earlier periods and have less periods that has zero effective demands. Let $t_{1}$ be the first period index that has zero effective demand in $\delta_{t}$. Therefore, another set $\delta^{'}_{t}$ is computed below:
% \begin{equation}
% \delta^{'}_{t}=
% \begin{cases}
% \delta_{t}\quad & t\neq t_{1}, t=m,m+1,\dots ,n\\
% 1\quad & t=t_{1}, t=m,m+1,\dots ,n
% \end{cases}
% \end{equation}

% Another sub-linear problem Model P-sub4 is computed, its result is  compared with that of Model P-sub3, and we select the one which gives maximum $BB_{m,n}$. 
% \\

% \textbf{Model P-sub4}
% \begin{alignat}{2}
%  &\text{Max}\quad && BB_{m,n} =\max\{B_{n}-B_{m-1}\}\tag{\ref{eq:BBmn}}\\
% &\text{s.t.}   && \text{for }t=m,m+1,\dots ,n  \nonumber\\
% &                 &&\eqref{con1:setupcapital}\text{--}\eqref{con1:ed0}\\ 
% &              && Ed_{t}=\begin{cases}
% d_{t}-\beta (Ed_{t-1}-v_{t-1})\qquad &\text{if }\delta^{'}_{t}=1\\
% 0\qquad &\text{if }\delta^{'}_{t}=0
% \end{cases}\\
% &              && d_{t}-\beta (Ed_{t-1}-v_{t-1})<0\qquad \text{if }\delta^{'}_{t}=0
% \end{alignat}

If Model P-sub1, Model P-sub2 and Model P-sub3 all do not have a feasible solution, we deem $BB_{m,n}$ does not have a feasible solution and set $v_{t}=0$ ($t=m,m+1,\dots,n$). The relation of $v_{t}$ with $w_{t}$ is provided by Eq. \eqref{eq:wtvt}. The relation of $v_{t}$ with $y_{t}$ is given by Eq. \eqref{relation:Xv}.
\begin{equation}
y_{t}=\begin{cases}
\sum\nolimits_{j=t_{i}}^{t_{i+1}-1}v_{j},\quad &t=t_{i},~i=1,2,\dots,k,\\
0, &t\neq t_{i},~i=1,2,\dots,k.
\end{cases}\label{relation:Xv}
\end{equation}

\subsection{Heuristic techniques in recursion and adjustments}
It is time consuming and complex to enumerate all the possible production cycles in a production round. Therefore, when customer goodwill loss rate is zero, we use one production cycle in a production round; when goodwill loss rate is not zero, we use at most two production cycles in a production round for computation.  For a certain period $t+1$ and given production plan from period 1 to period $t$, two situations are considered in computation of capital increment during a production round when goodwill loss rate is not zero: if there exists no production cycle before $t+1$, we compute only one production cycle beginning with period $t+1$ as a production round; if there exists production cycles before $t+1$, we view the nearest previous production cycle and the production cycle beginning with $t+1$ together as a production round, and make computations.

After computation of capital increments during production rounds, we can get an approximated production plan from period 1 to any period $n$ ($1\leq n\leq T$). Based on Corollary \ref{corollaryAdjust}, we make heuristic adjustments to this production plan. Three situations are considered for this adjustment, which are shown by Figure \ref{fig:adjust123}. In Figure \ref{fig:adjust123}, period  $n$'s production cycle beginning at period $t+1$.
\begin{figure}[!ht]
\centering
\subfigure[Heuristic adjustment 1.]{\label{adjust1}
\includegraphics[scale=0.8]{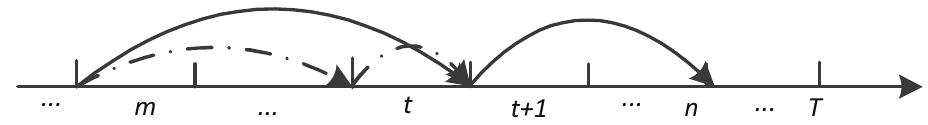}}

\subfigure[Heuristic adjustment 2.]{\label{adjust2}
\includegraphics[scale=0.8]{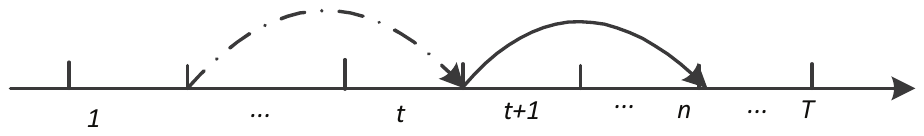}}~~~~  
\subfigure[Heuristic adjustment 3.]{\label{adjust3}
\includegraphics[scale=0.8]{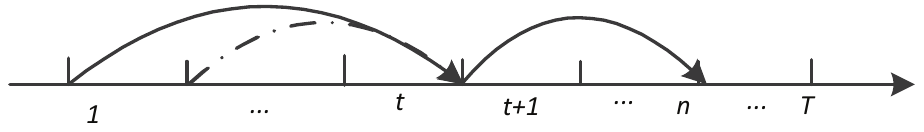}}
\caption{Heuristic adjustments for a known production plan from period 1 to $n$.}\label{fig:adjust123}
\end{figure}

\begin{itemize}
\item Figure \ref{adjust1} means we make adjustments to the first production cycle in the production round $m$ to $n$: dividing the first production cycle into two cycles by enumerating all new production launching periods between period $m$ to period $t$; recomputing $BB_{m,n}$, and selecting the one which gives maximum capital increment.

\item Figure \ref{adjust2} means sometimes it's better to launch a new production cycle before period $t+1$ when there exists no production cycles before it: enumerating all new production launching period between period $1$ to period $t$ as $m$; recomputing $BB_{m,n}$ and selecting the optimal one.

\item Figure \ref{adjust3} means sometimes it's better to launch production later if the first production cycle include period 1: enumerating all new production launching period between period $1$ to period $t$ as $m$;  recomputing $BB_{m,n}$, and selecting the optimal one that can give maximum capital increment.
\end{itemize}

In computing $BB_{m,n}$ for the three heuristic adjustments, a new linear constraint is added to the sub-linear problems: $w_{n}'\leq w_{n}$, which means demand shortage at period $n$ after the adjustment should be not higher than its original value before the adjustment.

If goodwill loss rate is zero, Corollary \ref{corollaryAdjust} is not necessary and there is no need for the adjustments above. After the recursion of $BB_{m,n}$ till the final period $T$, a production plan for the the whole planning horizon is obtained. Backward from period $T$ to period 1, check if it satisfies the criteria of Corollary \ref{corollaryInventory} and make production adjustments. 

\subsection{Computation Algorithm}

Based on functional equation \eqref{eq:recurecomputation}, sub-linear problems and heuristic techniques, we propose a forward recursive algorithm with heuristic adjustment algorithm (FRH) to solve Model P.
\\

\small\textbf{Algorithm FRH for Model P}
\vspace{3mm}

\small\underline{\emph{initialization:}}~~$t=1$, $m=1$, $1\times T$ zero matrices $x$, $y$, $B^{\ast}$, $T\times T$ zero matrices $BB$.
\setlength{\parskip}{5 pt}

\small\underline{\emph{Step 1:}}~~For $n=t, t+1, \dots, T$, select the production round beginning at $m$ and end at $n$, compute $BB_{m,n}$ and record its value in $BB(t,n)$.  

\small\underline{\emph{Step 2:}}~~Compute $B_{t}^{\ast}$ according to Eq. \eqref{eq:recurecomputation}, and obtain the production plan from period 1 to period $t$: $x (1:t)$, $y (1:t)$, $w (1:t)$.  

\small\underline{\emph{Step 3:}}~~Check if the present production plan from period 1 to $t$ meets the three adjustment situations shown by Figure \ref{fig:adjust123}. If meeting the adjustment criteria, make adjustments and update $x (1:t)$, $y (1:t)$, $w (1:t)$, $BB(t,n:T)$.
 
\small\underline{\emph{Step 4:}}~~$t=t+1$, update $m$ and repeat Step 1 $\sim$ Step 3 until $t=T$.

\small\underline{\emph{Step 5:}}~~Check if production plan meets Corollary \ref{corollaryInventory}, if meeting the criteria, make plan adjustments. Obtain final production plan: $x (1:T)$, $y (1:T)$, $w (1:T)$ and final capital $B^{\ast}_{T}$.
\normalsize\setlength{\parskip}{6 pt}

Flow char of our algorithm is shown by Figure \ref{fig:flowchart}.
\begin{figure}[h]
\centering\includegraphics[height=10cm]{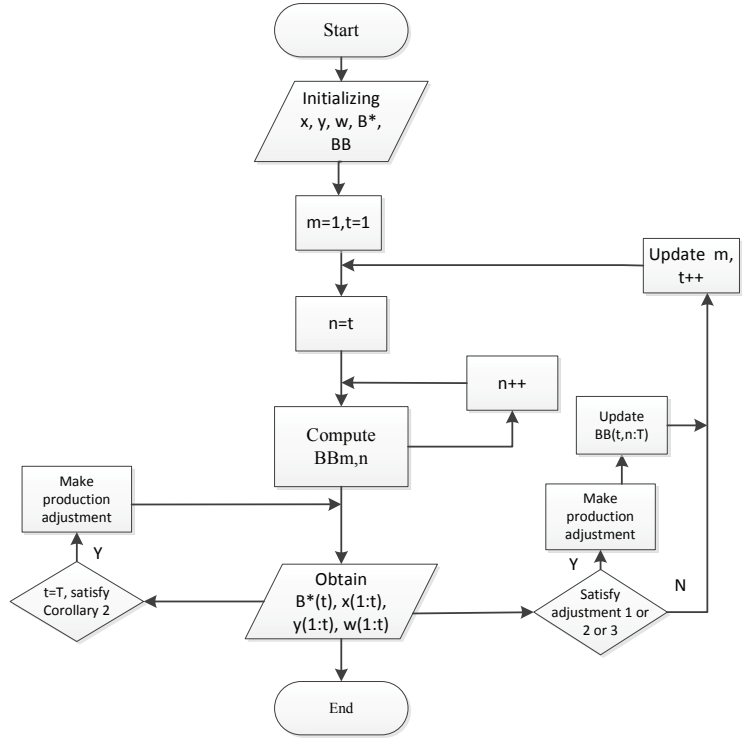}
\caption{Flow chart of our algorithm.}\label{fig:flowchart}
\end{figure}
\normalsize\setlength{\parskip}{0 pt}

\subsection{Computation complexity of our algorithm}
During recursion and heuristic adjustments, when customer goodwill loss rate is zero, there are $T(T+1)/2$ computations of $BB_{m,n}$ and $BB_{m,n}$ includes one sub-linear problem; when customer goodwill loss rate is not zero, in the best case, it is same as the situation when goodwill loss rate is zero: there are no heuristic adjustments meeting Corollary \ref{corollaryAdjust}, and computation of $BB_{m,n}$ requires only one sub-linear problem. In the worst case, there are $3T(T+1)/2$ computation of $BB_{m,n}$ in total: $T(T+1)/2$ computations of $BB_{m,n}$ at first, at most $T(T+1)/2$ computations of $BB_{m,n}$ for the heuristic adjustments, and $T(T+1)/2$ computations of $BB_{m,n}$ to update $BB(t,n:T)$; each computation of $BB_{m,n}$ includes 3 sub-linear problems.

Therefore, there are $T(T+1)/2$ computations of sub-linear problems in the best case and at most $9T(T+1)/2$ computations of sub-linear problems in the worst case. The computational complexity of our algorithm is $O(T^{2}\psi)$, where $\psi$ is the computational complexity of the algorithm for the sub-linear problems. 

Without integer variables, the sub-linear problems can be solved by polynomial interior point algorithm. For the common used polynomial interior point algorithm by \citet{karmarkar1984new}, $\psi$ is $O(T^{3.5}\mathcal{L})$, where $\mathcal{L}$ denotes the
total length of the binary coding of the input data. This is the reason why we remove integer variables from original mixed integer model and divide it into sub-linear problems. Therefore, total computational complexity of our algorithm is $O(T^{5.5}\mathcal{L})$, which is a polynomial algorithm.

%  (in the case of problems with integer variable), or $\ln(1/\epsilon)$ if the problem under consideration is solved to the
% relative accuracy $\epsilon$ (in the case of problems without integer variable). Details about complexity of different interior point algorithm could be found in \cite{potra2000interior}.  Obviously, for the linear problems with integer variables, $\mathcal{L}$ could be much larger than the case of linear problems without integer variables.

\section{Numerical analysis}
In this section, we first employ some numerical examples to show the influence of initial capital and loan, to the optimal production plan and final capital increment of a manufacturer, and then compare our algorithm with the business software CPLEX. In our numerical experiments, the linear programming algorithm for sub-linear problems is a function in MATLAB based on the paper of \citet{zhang1998solving24}, which is an interior point algorithm. The solution accuracy of the interior algorithm in MATLAB is controlled by the termination tolerance on the function, which we set as 0.0001\%. The number of maximum iterations for interior point algorithm is set to be 50. 

Our algorithm is coded in MATLAB 2016a, and run on a desktop computer with an Intel (R) Core (TM) i5-6500 CPU, at 3.20 GHz, 16GB of RAM, and 64-bit Windows 7 operating system. 

\subsection{Numerical examples about influence of capital flow to production plan}

Assume $T=12$, and goodwill loss rate $\beta=0.5$. The values of some other parameters are listed in Table \ref{tab:parametervalues}. 

\begin{table}[!ht]
\centering\small
\captionsetup{skip=2pt}
\caption{Parameter values}\label{tab:parametervalues}
\begin{tabular}[b]{*{13}{p{0.5cm}<{\raggedright}}}
\toprule
$p_{t}$  &21  &22  &20  &15  &10  &8   &5   &10  &18  &10  &14  &18\\
\midrule
$c_{t}$  &5   &13  &10  &10  &10  &10  &10  &10  &10  &10  &10  &10\\
$h_{t}$  &10  &5   &5   &5   &5   &5   &5   &5   &5   &5   &5   &5 \\
$s_{t}$  &100 &100 &100 &100 &100 &100 &100 &100 &100 &100 &100 &100\\
$d_{t}$  &30  &45  &50  &55  &45  &55  &90  &80  &90  &65  &80  &70 \\
\bottomrule
\end{tabular}
\end{table}
\normalsize

We solve the problem via our algorithm FRH. The solutions are optimal verified by CPLEX. When initial capital is 150 without loan, the optimal production plan is shown by Figure \ref{fig:example1}, in which the manufacturer could only launch two productions because of capital shortage. When initial capital is 200 without loan, the optimal production plan is shown by Figure \ref{fig:example2}. When initial capital is 200 with  loan quantity 300, loan length 3 periods and loan interest rate 10\%, a different optimal production plan is shown by Figure \ref{fig:example3}. Figures in Figure \ref{fig:example123} illustrate that quantity of initial capital and whether or not loan, does influence the optimal production plan of a manufacturer.

\begin{figure}[!ht]
\centering
\subfigure[Production plan with initial capital 150 without loan.]{\label{fig:example1}
\includegraphics[scale=0.8]{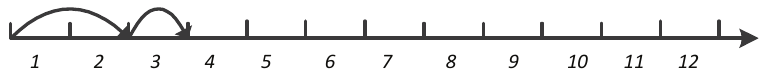}}
\subfigure[Production plan with initial capital 200 without loan.]{\label{fig:example2}
\includegraphics[scale=0.8]{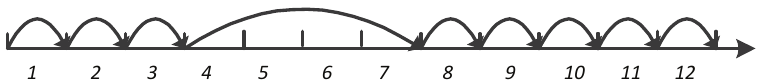}}~~~~
\subfigure[Production plan with initial capital 200 with loan.]{\label{fig:example3}
\includegraphics[scale=0.8]{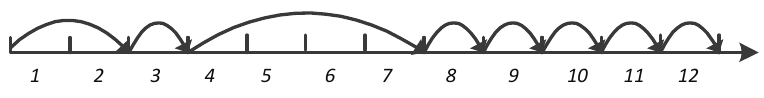}}
\caption{Optimal production plan for different capital situations.}\label{fig:example123}
\end{figure}

Without loan, for different initial capital, maximum final capital increments are displayed by Figure \ref{fig:capitalchanges1}. With fixed quantity of initial capital 200, fixed quantity of loan 300, loan length 3 periods, for different loan interest rates, maximum final capital increments is displayed by Figure \ref{fig:capitalchanges2}.

\begin{figure}[!ht]
\centering
\subfigure[Capital increment with different initial capital.]{\label{fig:capitalchanges1}
\begin{tikzpicture}
\pgfplotsset{xticklabel style={
        /pgf/number format/fixed,
}}
\begin{axis}[xlabel=Initial capital,
    ylabel=Final capital increment, xtick=data,font=\footnotesize,scale=0.8]
\addplot table [x=Initial capital, y=Capital increment1, col sep=comma] {mydata.csv};
\end{axis}
\end{tikzpicture}}
~~~~~~
\subfigure[Capital increment with different loan interest rate.]{\label{fig:capitalchanges2}
\begin{tikzpicture}
\pgfplotsset{xticklabel style={
        /pgf/number format/fixed,
}}
\draw [dashed] (0,2.5) -- (5.5,2.5);
\begin{axis}[xlabel=Loan interest rate,
    ylabel=Final capital increment, xtick=data, font=\footnotesize,scale=0.8]
\addplot table [x=Interest rate, y=Capital increment2, col sep=comma]
{mydata.csv};
\end{axis}
\end{tikzpicture}}
\caption{Changes of final capital increment for different initial capital and loan interest rates.}\label{fig:capitalchanges}
\end{figure}
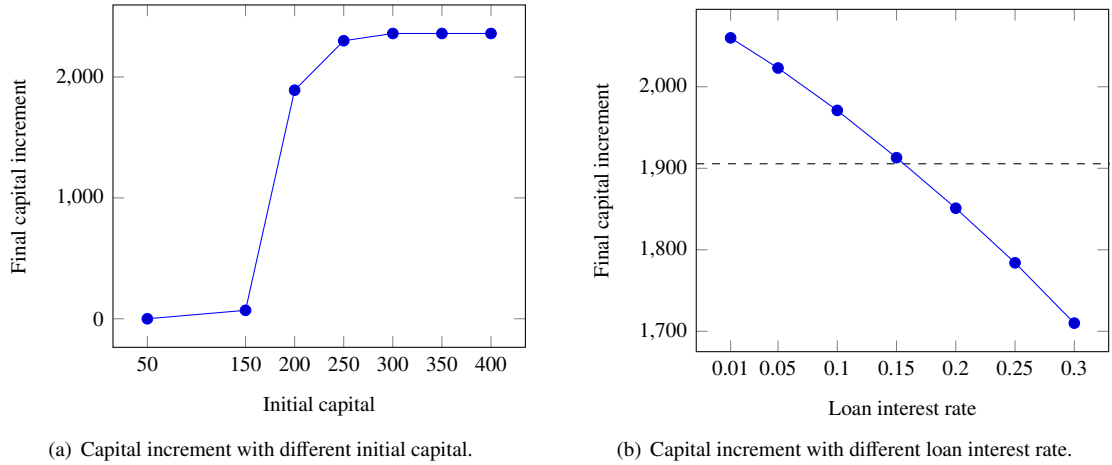

% \begin{figure}[!ht]
% \centering
% \subfigure[Capital increment with different initial capital]{\label{fig:capitalchanges1}
% \includegraphics[scale=0.5]{capitalchanges1.eps}}~~
% \subfigure[Capital increment with different interest rates]{\label{fig:capitalchanges2}
% \includegraphics[scale=0.5]{capitalchanges2.eps}}
% \caption{Changes of final capital increment for different initial capital and interest rates}\label{fig:capitalchanges}
% \end{figure}

In Figure \ref{fig:capitalchanges2}, the dashed line represents the maximum capital increment without loan. Figure \ref{fig:capitalchanges1} shows if capital is not sufficient, more initial capital will bring more final capital increment; if capital is sufficient, maximum final capital increment for a manufacturer is stable. Figure \ref{fig:capitalchanges2} shows loan is helpful for a manufacturer if interest is low; but if interest rate is too high, final capital increment decreases and it is better for the manufacture not to loan. Therefore, the numerical examples above demonstrate initial capital availability as well as loan interest rate can substantially influence the operational decisions for a manufacturer. 

\subsection{Comparison of our algorithm with some other heuristics}
To the best of our knowledge, although there are many heuristic algorithms in the literature dealing with capacity constrained lot sizing problems, those algorithms are not suitable for solving capital flow constrained lot sizing problem like ours. Solutions obtained by those algorithms are not feasible by the definition of capital flow constrains in this paper. Compared with traditional capacity constraints, capital flow constraints are stronger constraints which require initial capital of each period is above this period's total production cost, end-of-period capital of each period is above zero, and capital flow is related with many parameters like selling price, interest rate, etc. 

In terms of the comparison of our algorithm with meta heuristics, we attempt to solve our capital flow constrained lot sizing problem with some meta heuristics: genetic algorithm and simulated annealing algorithm. However, because of the capital flow constraints and other constraints, it is difficult for both genetic algorithm and simulated annealing algorithm to obtain a feasible solution even for a small numerical case. Therefore, we omit the comparison of our algorithm with other heuristic algorithms.

\subsection{Comparison of our algorithm with CPLEX on randomly generated problems}

We test our algorithm on a large set of randomly generated problems with CPLEX 12.6.2. The solution accuracy of CPLEX is controlled by the number of iterations, CPU seconds and the termination tolerance, which are set as 750,000, 18,000, and 0.0001\%, respectively. The randomized scheme of test problem generation is similar to Aksen's work \citep{aksen2007loss19}, and is presented in Table \ref{tab:generate tests}. 

\begin{table}[!ht]
\centering\small
\caption{Randomized generation scheme for the test problems}\label{tab:generate tests}
\begin{tabular}{m{3.5cm}m{1.2cm}m{6cm}}
\toprule
\emph{Parameter} &\emph{No.}  &\emph{Values}\\
\midrule
Planning horizon  &6 & $T=$ 12, 24, 36, 48, 60, 72\\
\specialrule{0em}{2pt}{2pt} 
\multirow{3}*{Demand distribution} &\multirow{3}*{3} &1. \emph{Exponential} with $\mu=150$ \\
                &          &2. \emph{Normal} with $\mu=150$, $\delta^{2}=1600$\\
                &          &3. \emph{Discrete Uniform} [30, 270, 10]\\
\specialrule{0em}{2pt}{2pt} 
Unit production cost &\multirow{2}*{2} &1. Both constant: $c_{t}=13$ and $h_{t}=1$\\
Unit holding cost              &          &2. Both seasonally varying\\
\specialrule{0em}{2pt}{2pt} 
\multirow{3}*{Selling price}  &\multirow{3}*{2} &1. Discrete uniformly distributed \\&&between
[15 25] with increments of 5\\   
&                   &2. Seasonally varying\\
\specialrule{0em}{2pt}{2pt} 
\multirow{2}*{Initial capital}  &\multirow{2}*{2}   
                  &1. $B_{c}=s_{1}+c_{1}(d_{1}+d_{2})$\\
&                   &2. $B_{c}=s_{1}+c_{1}(d_{1}+d_{2}+d_{3})$\\
\specialrule{0em}{2pt}{2pt} 
\multirow{3}*{Initial loan}  &\multirow{3}*{2}   
                  &1. $B_{L}=0$\\
&                   &2. $B_{L}=2000$, with loan length $T_{L}=6$\\
&&~~and loan rate $r=5\%$\\
\specialrule{0em}{2pt}{2pt} 
\multirow{3}*{Goodwill loss rate } &\multirow{3}*{3} &1. $\beta=0$ \\
                &          &2. $\beta=10\%$\\
                &          &3. $\beta=50\%$\\
\specialrule{0em}{2pt}{2pt}                
Setup cost  &1 &Constant: $s_{t}=1000$\\
\bottomrule
\end{tabular}
\end{table}
\normalsize

Since capital and goodwill loss rate can influence optimal production plan, two initial capital, two initial loan and three goodwill loss rates are set for our experiments, while these three parameters in Aksen's 360 test cases \citep{aksen2007loss19} are fixed or not included. As for the initial capital, $B_{c}=s_{1}+c_{1}(d_{1}+d_{2})$ guarantees the manufacture has enough capital for the production of first two periods; and $B_{c}=s_{1}+c_{1}(d_{1}+d_{2}+d_{3})$ guarantees the production of first three periods. There are 864 numerical cases for testing in total. Experimental results for different periods are shown by Table \ref{tab:resultsforperiods}.

\begin{table}[!ht]
\centering\small
\caption{Performance of our algorithm FRH compared with CPLEX for different periods}\label{tab:resultsforperiods}
\begin{tabular}{cccccrr}
\toprule
$T$  &\emph{\tabincell{l}{Num of\\cases}}  &\emph{\tabincell{c}{Num of\\n-opt cases}}  &\emph{\tabincell{c}{Average\\deviation}}  &\emph{\tabincell{c}{Maximum\\deviation}}&\emph{\tabincell{c}{ Avg.FRH\\time ($s$)}}   &\emph{\tabincell{c}{Avg.CPLEX\\time ($s$)}}   \\
\midrule
12 &144 &4 &0.07\% &4.56\% &0.85&0.18\\
24 &144  &5 &0.02\% &4.38\%&3.54&0.33\\
36 &144  &6 &0.03\% &1.29\%&9.06&1.83\\
48 &144 &11 &0.08\% &3.74\%&18.12&136.67\\
60 &144 &15 &0.02\% &0.52\%&32.38&267.87\\
72 &144  &16 &0.05\% &1.59\%&55.76&1062.80\\
\specialrule{0em}{2pt}{2pt} 
\emph{General} &844 &57 &0.05\% &4.56\% &19.95&224.19\\
\bottomrule
\end{tabular}
\end{table}
\normalsize

As shown in Table \ref{tab:resultsforperiods}, our algorithm performs well in the 844 test cases. There are only 57 cases that our algorithm doesn't get optimal solutions. Although there still exist some extreme cases with a maximum deviation 4.56\% that our heuristic adjustment could not reach optimal, it could obtain optimal solutions in most cases and average deviation is 0.05\%. Another finding of the experiment not shown by Table \ref{tab:resultsforperiods} but also should be noted is that, when goodwill loss rate is zero and unit variable production costs are equal, our algorithm all obtain optimal solutions, which validates Theorem \ref{theorem1}.  

In terms of computation time, CPLEX runs faster than our algorithm for small-size problems. However, when problem size grows large, average computation time for CPLEX increase rapidly and is much larger than our algorithm. This is because when $T$ reaches 48, there are some cases that take maximum running time for CPLEX to stop iteration, which boost the average computation time. For the 864 numerical cases, average computation time of FRH algorithm is 19.95s while average computation time for CPLEX is 224.19s.  Therefore, our algorithm is suitable for solving large-size problems.

\subsection{Factors affecting the performance of our algorithm}

In the next stage of experiment, we redesign the generation scheme of the test problems to investigate the influence of parameter values to the performance of our algorithm. In order to save computation time for comparison, we set production horizon length $T$ fixed to be 12, initial loan $B_{L}$ fixed to be 2000, and loaning length $L$ fixed to be 6 in this stage of testing, setup cost $s$ is also set fixed to be 1000. For other parameters, each has 2 generation modes: high fluctuations and low fluctuations with normal distribution, or high values and low values. Details of the generation scheme is displayed by Table \ref{tab:generate tests2}.

\small
\renewcommand\arraystretch{1.5}
\begin{longtable}{lllll}
\caption{Randomized generation scheme in the second stage of testing}\label{tab:generate tests2}\\
\toprule
    &\emph{Low value (fluc.)}  &\emph{High value(fluc.)}\\
\midrule
Demand $D$&$\mu=150, \delta=10$ &$\mu=150, \delta=50$\\
Unit production cost $c$ &$\mu=13, \delta=1$ &$\mu=13, \delta=5$\\
Unit holding cost $h$ &$\mu=5, \delta=0.5$ &$\mu=5, \delta=2.5$\\
Selling price $p$ &$\mu=20, \delta=1$  &$\mu=20, \delta=5$\\
Initial capital $B_{c}$ &$s_{1}+c_{1}(d_{1}+d_{2})$
&$s_{1}+c_{1}\sum\nolimits_{i=1}^{5}d_{i}$\\
Interest rate $r$ &2\% &5\%\\
Goodwill loss rate $\beta$ &10\% &50\%\\
\bottomrule
\end{longtable}
\normalsize

For each combination of those parameters, we generate 10 numerical cases. Therefore, there are $2^{7}\times 10=1280$ cases for testing. Experimental results in this stage are presented by Table \ref{table: results2}.

\renewcommand\arraystretch{1}
\begin{table}[!ht]
\centering\small
\caption{Pivot table of the results in the second stage of testing}\label{table: results2}
\begin{tabular}{m{2.3cm}*{4}{m{1.8cm}}}
\toprule
   &\emph{\tabincell{l}{Num of\\cases}}  &\emph{\tabincell{l}{Num of\\deviation}}  &\emph{\tabincell{l}{Average\\deviation}}   &\emph{\tabincell{l}{Maximum\\deviation}}\\
\midrule
Demand &&&&\\
\specialrule{0em}{1pt}{1pt}
\emph{low fluc.} &640 &38 & 0.14\% & 5.91\%\\
\emph{high fluc.}  &640 &57 & 0.09\% &5.68\%\\
\specialrule{0em}{2pt}{2pt}
Unit prod. cost &&&&\\
\specialrule{0em}{1pt}{1pt}
\emph{low fluc.}&640 &52 & 0.17\% &5.91\%\\
\emph{high fluc.}&640 &44 & 0.06\% &4.48\%\\
\specialrule{0em}{2pt}{2pt}
Unit hold. cost &&&&\\
\specialrule{0em}{1pt}{1pt}
\emph{low fluc.}&640 &53 & 0.13\% &5.91\%\\
\emph{high fluc.}&640 &43 & 0.09\% &5.68\%\\
\specialrule{0em}{2pt}{2pt}
Selling price &&&&\\
\specialrule{0em}{1pt}{1pt}
\emph{low fluc.}&640 &57 & 0.14\% &5.91\%\\
\emph{high fluc.}&640 &39 & 0.09\% &4.48\%\\
\specialrule{0em}{2pt}{2pt}
Initial capital &&&&\\
\specialrule{0em}{1pt}{1pt}
\emph{low}&640 &57 & 0.18\% &5.91\%\\
\emph{high}&640 &39 & 0.05\% &2.48\%\\
\specialrule{0em}{2pt}{2pt}
Interest rate &&&&\\
\specialrule{0em}{1pt}{1pt}
\emph{low}&640 &46 & 0.12\% &5.68\%\\
\emph{high}&640 &50 & 0.09\% &5.91\%\\
\specialrule{0em}{2pt}{2pt}
Goodw. loss rate &&&&\\
\specialrule{0em}{1pt}{1pt}
\emph{low}&640 &46 & 0.03\% &2.78\%\\
\emph{high}&640 &50 & 0.19\% &5.91\%\\
\specialrule{0em}{3pt}{3pt}
General &1280 &96 &0.11\% &5.91\%\\
\bottomrule
\end{tabular}
\end{table}

From table \ref{table: results2}, of all the 1280 numerical cases, there are 96 numerical cases that our algorithm can't reach optimal with maximum deviation error 5.91\% and average deviation 0.11\%. we also find that for goodwill loss rate and initial capital, maximum deviation and average deviation between high and low values differ substantially. It seems goodwill loss rate and initial capital play a main role in affecting the performance of our algorithm. To consolidate this conclusion, we apply stepwise linear regression analysis by SPSS to all the 96 numerical cases that our algorithm has deviations. We set deviation as dependent variable, all the seven parameters as independent variables, confidence interval is 95\%. Analysis of variance (anova) is presented by Figure \ref{fig:anova} and excluded variables by stepwise linear regression is given by Figure \ref{fig:excludedPara}.

\begin{figure}[!ht]
\centering
\includegraphics[scale=0.8]{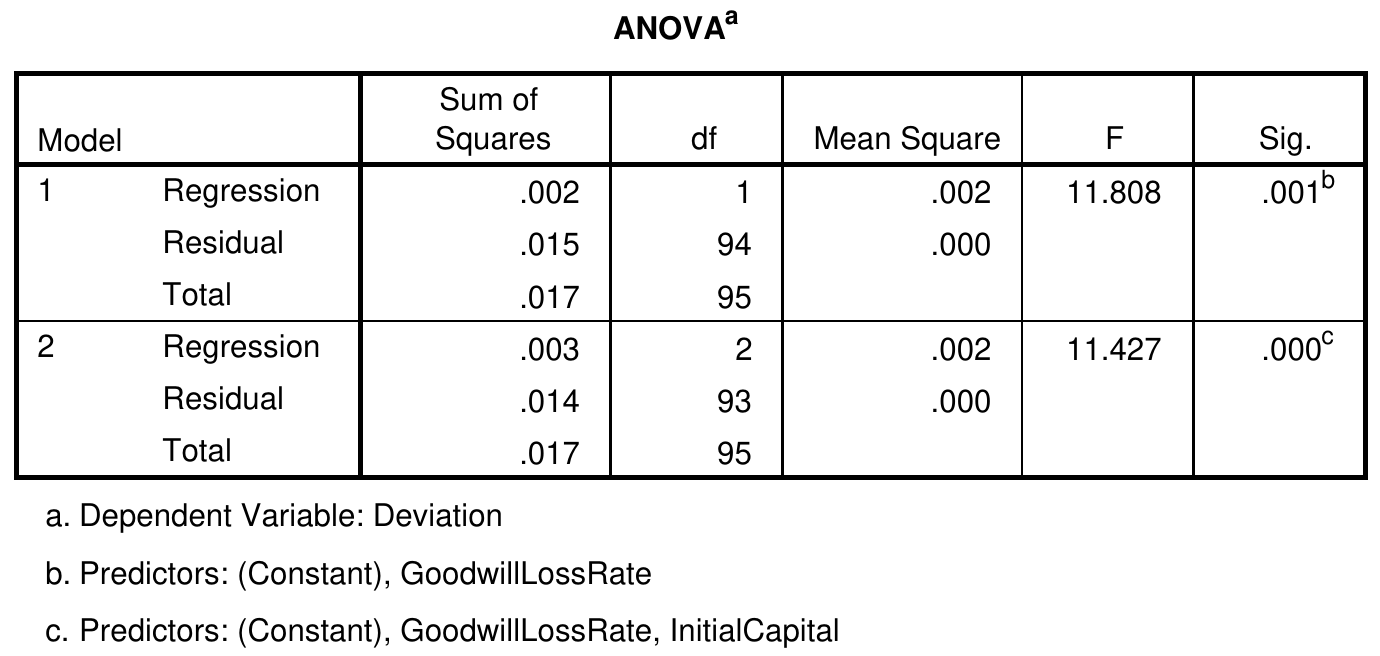}
\caption{Anova of parameters in stepwise linear regression.}\label{fig:anova}
\end{figure}
\begin{figure}[!ht]
\centering
\includegraphics[scale=0.8]{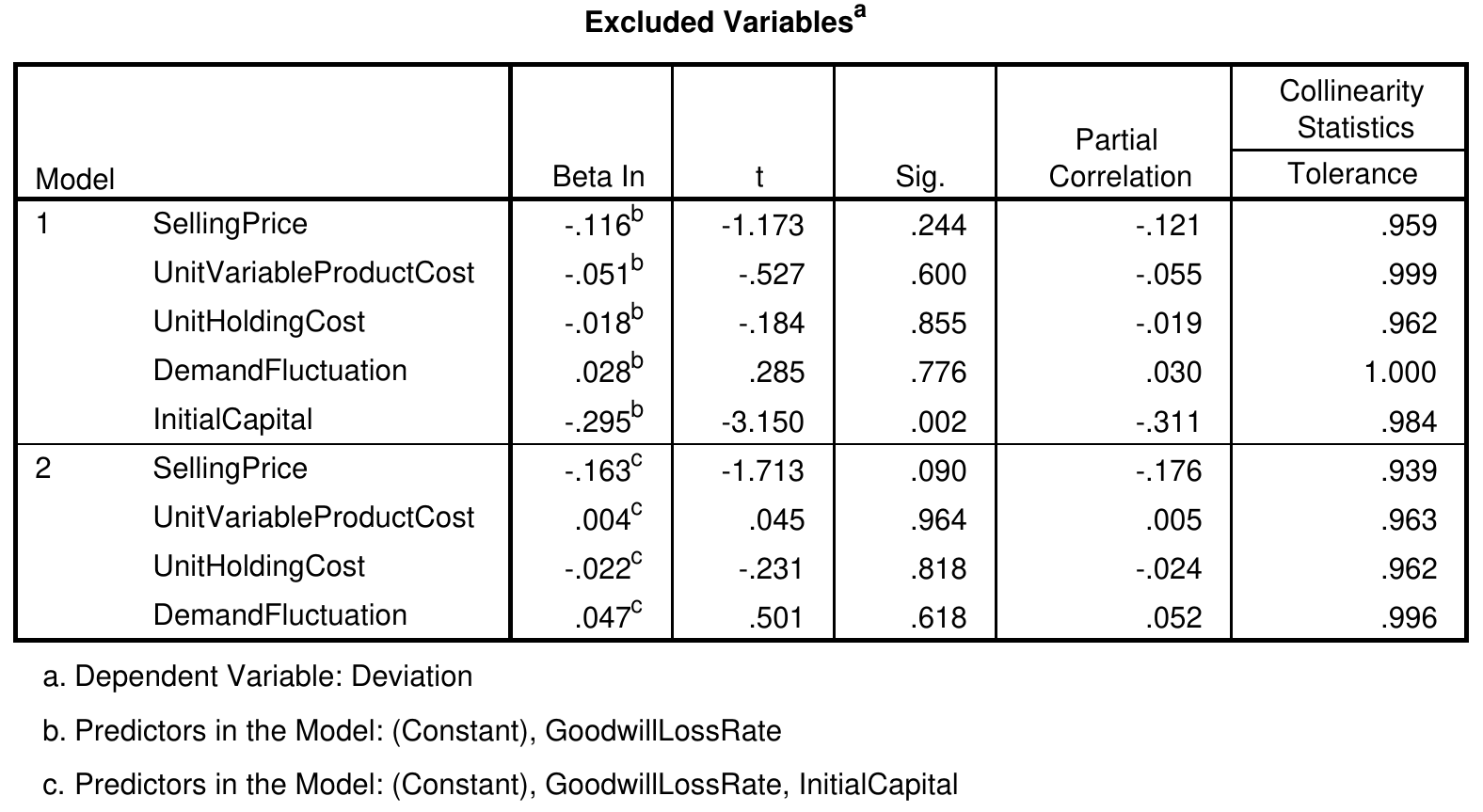}
\caption{Excluded parameters in stepwise linear regression.}\label{fig:excludedPara}
\end{figure}

Figure \ref{fig:anova} shows goodwill loss rate affects deviation the most (significance value $0.001<0.05$); initial capital and goodwill loss rate together have significant influence to the deviation error (significance value $0.000<0.05$), while other parameters are excluded from regression as shown by Figure \ref{fig:excludedPara}. This coincides with the finding in Table \ref{table: results2}. The reason could be: if initial capital is low and goodwill loss rate is high, it is more difficult for the heuristic techniques to adjust original solution to optimal. However, experiments in the two stages demonstrate our algorithm can reach optimal in over 90\% cases and average deviation of our algorithm is rather low; moreover, when goodwill loss rate is zero and unit variable production are equal, our algorithm can definitely obtain optimal solutions.

\section{Conclusions}

Capital shortage is a key factor affecting the growth of many small and medium enterprises. However, capital flow constraints have not been taken into consideration by many lot sizing works. Previous methods such as Wagner-Whitin algorithm \citep{wagner1958dynamic1} and Aksen algorithm \citep{aksen2003single7,aksen2007loss19} for lot sizing problems can not obtain feasible solutions when considering capital flow constraints under the assumptions in our paper. 

We formulate a mathematical model for the lot sizing problem with capital flow constraints. Loss of goodwill and loan are also introduced in our problem. Based on the mathematical properties of the problem, we develop a forward recursion algorithm with heuristic adjustments. When unit variable productions costs are equal and goodwill loss rate is zero, our algorithm can obtain optimal solutions. Under other situations, its average deviation error is rather low. It is suitable to solve large-size problems for its computational efficiency. We also find initial capital availability and loan interest rate can affect a manufacturer's optimal lot sizing decisions

Future research could extend in several directions: first is considering the multi-item model or the stochastic lot sizing problems with  capital flow constraints; second, other financial behaviors, such as trade credit, inventory financing and factoring business could also be taken into account in the lot sizing problem.

% \section*{Conflict of interests}
% The authors declare that there is no conflict of interests
% regarding the publication of this paper.

% \section*{Acknowledgments}
% The research of Zhen Chen is supported by NSFC under
% 71271010 and 71571006.

\section*{Reference}

% \par
% \vspace{10mm}
% \noindent\large\textbf{References}
% \vspace{3mm}

%\renewcommand\refname{Reference}
\bibliographystyle{jors}
\small
\setlength{\bibsep}{0.5ex} % 缩小参考文献间的垂直间距
\bibliography{literature.bib}

\end{document}